\documentclass[12pt]{article}
\usepackage[round,colon,authoryear]{natbib}
\usepackage{bibentry}
\RequirePackage[OT1]{fontenc}
\RequirePackage{amsthm,amsmath,amsfonts,amssymb,{mathtools}}
\RequirePackage[colorlinks]{hyperref}
\usepackage{bbm}
\hypersetup{
 colorlinks=false,
 citecolor=Violet,
 linkcolor=Red,
 urlcolor=Blue}
\RequirePackage{epsfig, graphicx, color}
\RequirePackage{latexsym}
\RequirePackage{psfrag}
\usepackage{multirow}
\usepackage{diagbox}
\usepackage{fullpage}
\usepackage{epsf,epsfig,subfigure}
\usepackage{graphics,graphicx}
\usepackage{enumerate}
\usepackage{epsfig}

\newtheorem{lem}{Lemma}
\newtheorem{thm}{Theorem}

\def\calM{\mathcal{M}}
\def\E{\mathbb{E}}
\def\P{\mathbb{P}}
\def\R{\mathbb{R}}
\def\bE{\mathbf{E}}
\def\ba{\mathbf{a}}
\def\bb{\mathbf{b}}
\def\bc{\mathbf{c}}
\def\bh{\mathbf{h}}

\def\bT{\mathbf{T}}
\def\bx{\mathbf{x}}
\def\bX{\mathbf{X}}

\def\bA{\mathbf{A}}
\def\bu{\mathbf{u}}
\def\bv{\mathbf{v}}
\def\bw{\mathbf{w}}
\def\vec{\mathrm{vec}}
\def\calN{\mathcal{N}}

\begin{document}
\title{Characterizing Spatiotemporal Transcriptome of Human Brain via Low Rank Tensor Decomposition$^\ast$}

\date{(\today)}

\author{Tianqi Liu$^\dag$, Ming Yuan$^\ddag$, and Hongyu Zhao$^\S$\\
$^{\dag,\S}$Yale University and $^\ddag$Morgridge Institute for Research\\
 and\\
$^\ddag$University of Wisconsin-Madison}

\footnotetext[1]{Address for Correspondence: Ming Yuan, Department of Statistics, University of Wisconsin-Madison, 1300 University Avenue, Madison, WI 53706.}
\footnotetext[3]{
The research of Ming Yuan was supported in part by NSF FRG Grant DMS-1265202, and NIH Grant 1-U54AI117924-01.}
\footnotetext[4]{The research of Hongyu Zhao was supported in part by NIH grants GM59507 and GM122078.}
\maketitle

\newpage

\begin{abstract}
Spatiotemporal gene expression data of the human brain offer insights on the spatial and temporal patterns of gene regulation during brain development. Most existing methods for analyzing these data consider spatial and temporal profiles separately with the implicit assumption that different brain regions develop in similar trajectories, and that the spatial patterns of gene expression remain similar at different time points. Although these analyses may help delineate gene regulation either spatially or temporally, they are not able to characterize heterogeneity in temporal dynamics across different brain regions, or the evolution of spatial patterns of gene regulation over time. In this article, we develop a statistical method based on low rank tensor decomposition to more effectively analyze spatiotemporal gene expression data. We generalize the classical principal component analysis (PCA) which is applicable only to data matrices, to tensor PCA that can simultaneously capture spatial and temporal effects. We also propose an efficient algorithm that combines tensor unfolding and power iteration to estimate the tensor principal components, and provide guarantees on their statistical performances. Numerical experiments are presented to further demonstrate the merits of the proposed method. An application of our method to a spatiotemporal brain expression data provides insights on gene regulation patterns in the brain.
\end{abstract}
\newpage

\section{Introduction}
\label{sec:intro}

Principal component analysis (PCA) is among the most commonly used statistical methods for exploratory analysis of multivariate data \citep[e.g.,][]{pcabook}. By seeking a low rank approximation to the data matrix, PCA allows us to reduce the dimensionality of the data, and oftentimes serves as a useful first step to capture the essential features in the data. In particular, PCA has been widely used in analyzing gene expression data collected for multiple time points or across different biological conditions. See, e.g., \cite{alter00pca, wall01pca, yeung01pca}. While PCA is appropriate to analyze data matrices, data sometimes come in the format of higher order tensors, or multilinear arrays. In particular, our work here is motivated by characterizing the spatiotemporal gene expression patterns of human brain based on gene expression profiles collected from multiple brain regions of both developing and adult post-mortem human brains.

Human brain is a sophisticated and complex organ that contains billions of cells with different morphologies, connectivity and functions \citep[e.g.,][]{kandel2000principles}. Different brain regions have specific compositions of cell types expressing unique combinations of genes at different developmental periods. Recent advances in sequencing and micro-dissection technology have provided us new and powerful tools to take a closer look at this complex system. Many studies have been conducted in recent years to collect spatiotemporal expression data to identify spatial and temporal signatures of gene regulation in the brain, and gain insights into various biological processes of interest such as brain development processes, central nervous system formation, and brain anatomical structure shaping, among others. See, e.g., \cite{wen1998large, kang2011spatio, parikshak2013integrative, miller2014transcriptional, pletikos2014temporal, landel2014temporal, hawrylycz2015canonical}. 

The spatiotemporal expression data can be naturally modeled by a third order multilinear array, or tensor, with one index for gene, one for region, and another one for time. Because the classical PCA can only be applied to data matrices, previous analyses of such data often consider the spatial and temporal patterns separately. To characterize temporal patterns of gene expression, data from different regions are first pooled and treated as replicates, before applying PCA. Similarly, when extracting spatial patterns of gene expression, data from different time points are combined so that PCA could be applied. Such analyses have yielded some useful insights on the gene regulation in spatiotemporal transcriptome. See, e.g., \cite{lein2007genome, kang2011spatio}. But the data pooling precludes us from understanding the heterogeneity in temporal dynamics across different regions of the brain, or the evolution of spatial gene regulation patterns over time. There is a clear demand to develop statistical methods that can more effectively utilize the tensor structure of spatiotemporal expression data.

To this end, we introduce in this article a higher order generalization, hereafter referred to as tensor PCA, of the classical PCA to better characterize spatial and temporal gene expression dynamics. As in the classical PCA, we seek the best low rank orthogonal approximation to the data tensor. The orthogonality among the rank-one components is automatically satisfied by the classical PCA but is essential for our purpose. It not only ensures that the components can be interpreted in the same fashion as the classical PCA, but also is necessary for the low rank approximation to be well-defined. Unlike in the case of matrices, low rank approximations to a higher order tensor without orthogonality is ill-posed and the best approximation may not even exist \citep[e.g.,][]{silvalim08}. However, even with orthogonality, low rank approximations to a higher order tensor is still in general NP hard to compute \citep[e.g.,][]{hillar13NP}. Heuristic or approximation algorithms are often adopted, and they often lead to suboptimal statistical performances \citep[e.g.,][]{montanari2014statistical}. It is an active area of research in recent years to achieve a balance between computational and statistical efficiency when dealing with higher order tensors. For our purposes, we propose an efficient algorithm that combines tensor unfolding and power iteration to compute the principal components under the tensor PCA framework. We show that our estimates not only are easy to compute but also attain the optimal rate of convergence under suitable conditions.

Numerical experiments further demonstrate the merits of our proposed method. In addition, we applied our method to the spatiotemporal expression data from \cite{kang2011spatio}, and found that the proposed tensor PCA approach can effectively reduce the dimensionality of the data while preserving inherent structure among the genes. In particular, through clustering analysis, we show that tensor PCA reveals interesting relationship between gene functions and the spatiotemporal dynamics of gene regulation. To fix ideas, we focus on spatiotemporal expression data in this paper. Our methodology, however, is also readily applicable to other settings where data are in the form of tensor.

The rest of the article is organized as follows. Section \ref{sec:meth} introduces the proposed tensor PCA methodology. Section \ref{sec:sim} reports the result from simulation studies. Section \ref{sec:real} presents an application of the proposed methodology to a spatiotemporal brain gene expression data set. Finally, we conclude with some remarks and discussions by Section \ref{sec:dis}. All proofs are relegated to Section \ref{sec:proof}.

\section{Methodology}
\label{sec:meth}

Denote by $x_{gst}$ an appropriately normalized and transformed expression measurement for gene $g$, in region $s$, at time $t$, where $g=1,\ldots, d_G$, $s=1,\ldots, d_S$, and $t=1,\ldots, d_T$, and $d_G$, $d_S$ and $d_T$ are the number of genes, regions, and time points, respectively. In many applications, we may also have replicate measurements so that $x_{gst}$ is a vector rather than a scalar. To fix ideas, we shall focus on the case where there is no replicate. Treatment of the more general situation is analogous albeit more cumbersome in notation.

\subsection{From classical PCA to tensor PCA}
As mentioned above, the classical PCA is often applied to estimate spatial and temporal patterns of gene regulation separately. Consider, for  example, inferring the spatial patterns of gene regulation. Let
$$
\bar{x}_{gs\cdot}={1\over d_T}\sum_{t=1}^{d_T} x_{gst},
$$
be the averaged expression measurements for gene $g$ in region $s$. The classical PCA then extracts the leading principal components, or equivalently the leading eigenvectors of $d_G\times d_S$ matrix $\bx_g:=(\bar{x}_{g1\cdot},\ldots, \bar{x}_{gd_S\cdot})^\top$. The principal components can also be interpreted through singular value decomposition of data matrix $(\bx_1,\ldots,\bx_{d_G})^\top$. Denote by $\bv_k:=(v_{k1},\ldots,v_{kd_S})^\top$ the $k$th leading principal component and $\bu_k:=(u_{k1},\ldots,u_{kd_G})^\top$ its normalized loadings, that is its $\ell_2$ norm $\|\bu\|=1$. Then, after appropriate centering, the observed expression measurements can be written as
\begin{equation}
\label{eq:classic}
\bar{x}_{gs\cdot}=\sqrt{d_G}\sum_{k=1}^r\lambda_k u_{kg} v_{ks}+\bar{\epsilon}_{gs},
\end{equation}
where $\lambda_1\ge \lambda_2\ge\cdots \lambda_r>0$ so that $\sqrt{d_G}\lambda_k$ is the $k$th largest singular value of the data matrix $(\bar{x}_{gs\cdot})_{1\le g\le d_G, 1\le s\le d_S}$, and the idiosyncratic noise $\bar{\epsilon}_{gs}$ are iid centered normal random variables. Note that, in (\ref{eq:classic}), the scaling factor $\sqrt{d_G}$ is in place to ensure that $\lambda_k^2$ (more precisely $\lambda_k^2+{\rm var}(\bar{\epsilon}_{gs})$) can also be understood as the $k$th largest eigenvalue of the covariance matrix of $(\bar{x}_{gs\cdot})_{1\le s\le d_S}$ when they are viewed as independent random vectors for $g=1,\ldots, d_G$.

Obviously, because of pooling measurements from different time points, the principal components extracted this way can only be identified with spatial patterns {\it averaged} over all time points. Therefore it is not able to capture spatial patterns that evolve over time. Similar problem also arises when we pool data from different regions and extract principal components for temporal patterns. In order to model the spatial and temporal dynamics jointly, we now consider a generalization of PCA to specifically account for the tensor structure of the expression data.

The expression data $\bX=(x_{gst})_{1\le g\le d_G,1\le s\le d_S,1\le t\le d_T}$ can be conveniently viewed as a third order tensor of dimension $d_G\times d_S\times d_T$. It is clear that the pooled data matrix
$$
(\bx_1,\ldots,\bx_{d_G})^\top=\bX\times_3 \left({1\over d_T}{\bf 1}_{d_T}\right),
$$
where ${\bf 1}_d$ is a $d$ dimensional vector of ones, and $\times_j$ between a tensor and vector stands for multiplication along its $j$th index, that is,
$$
(\bA\times_3 \bx)_{ij}=\sum_k A_{ijk}x_k.
$$
See, e.g., \cite{KoldarBader} for further discussions on tensor algebra. Instead of seeking a low rank approximation to the pooled data matrix, we shall work directly with the data tensor $\bX$. More specifically, with slight abuse of notation, we shall consider the following low rank approximation to $\bX$:
\begin{equation}
\label{eq:pca}
\bX=\sqrt{d_G}\sum_{k=1}^r\lambda_k \left(\bu_k\otimes \bv_k\otimes \bw_k\right)+\bE,
\end{equation}
where the eigenvalues $\lambda_1\ge\cdots\ge \lambda_r>0$, $\bu_k$s, $\bv_k$s and $\bw_k$s are orthonormal basis in $\R^{d_G}$, $\R^{d_S}$ and $\R^{d_T}$ respectively, and the $\bE=(e_{gst})$ is the residual tensor consisting of independent idiosyncratic noise following a normal distribution $N(0,\sigma^2)$. Here $\otimes$ stands for the outer product so that
$$
x_{gst}=\sqrt{d_G}\sum_{k=1}^r\lambda_k u_{kg} v_{ks}w_{kt}+e_{gst},\qquad \forall 1\le g\le d_G, 1\le s\le d_S, 1\le t\le d_T.
$$

Conceptually, model (\ref{eq:pca}) can be viewed as a natural multiway generalization of the model for the classical PCA. Similar to the classical PCA, such a tensor decomposition allows us to conveniently capture the spatial dynamics and temporal dynamics by $\bv_k$s and $\bw_k$s, respectively. The loading of each gene for a particular interaction of spatial and temporal dynamics is then represented by $\bu_k$s.

\subsection{Estimation for tensor PCA}
Clearly, any interpretation of the data based on the tensor PCA model (\ref{eq:pca}) depends on our ability to estimate the principal components $\bv_k$s and $\bw_k$s from the expression data $\bX$. Naturally, we can consider estimating them via maximum likelihood, leading to the problem of computing the best rank $r$ approximation to data tensor $\bX$. In the case of the usual PCA, such a task can be accomplished by applying SVD to the data matrix. But for the tensor PCA model, this is a more delicate issue because low rank approximation to a generic tensor could be hard to compute at least in the worst case. To address this challenge, we introduce here an approach that combines tensor unfolding and power iteration and show that we can estimate the tenor principal components in an efficient way, both computationally and statistically.

\subsubsection{Tensor unfolding}
A commonly used heuristic to overcome this problem is through tensor unfolding. In particular, in our case, we may collapse the second and third indices of $\bX$ to unfold into a $d_G\times (d_S\cdot d_T)$ matrix $\calM(\bX)$ by collapsing the second and third indices, that is,
$$
[\calM(\bX)]_{i,(j-1)d_T+k}=X_{ijk},\qquad \forall 1\le i\le d_G, 1\le j\le d_S, 1\le k\le d_T.
$$
It is clear that
$$
\calM(\bX)=\sqrt{d_G}\sum_{k=1}^r \lambda_k \bu_k\otimes \vec(\bv_k\otimes \bw_k)+\calM(\bE),
$$
where $\vec(\cdot)$ vectorizes a matrix into a vector of appropriate dimension. This suggests that $\{\vec(\bv_k\otimes \bw_k): 1\le k\le r\}$ are the top right singular vectors of $\E[\calM(\bX)]$ and can therefore be estimated by applying singular value decomposition to $\calM(\bX)$. Denote by $\sqrt{d_G}\widehat{\lambda}_k$ the $k$th leading singular value of $\calM(\bX)$, and $\widehat{\bh}_k$ its corresponding right singular vector. We can reshape $\widehat{\bh}_k$ into a $d_S\times d_T$ matrix $\vec^{-1}(\widehat{\bh}_k)$, that is
$$
[\vec^{-1}(\widehat{\bh}_k)]_{ij}=(\widehat{\bh}_k)_{(i-1)d_T+j},\qquad \forall 1\le i\le d_S, 1\le j\le d_T.
$$
An estimate of $\bv_k$ and $\bw_k$ can then be obtained by the leading left and right singular vectors, denoted by $\widehat{\bv}_k$ and $\widehat{\bw}_k$ respectively, of $\vec^{-1}(\widehat{\bh}_k)$. It turns out that this simple approach can yield a consistent estimate of $\lambda_k$s, $\bv_k$s and $\bw_k$s. More specifically, we have

\begin{thm}
\label{th:matricize}
There exists an absolute constant $C>0$ such that for any simple eigenvalue $\lambda_k$ ($1\le k\le r$) under the tensor PCA model (\ref{eq:pca}), if the eigen-gap
$$g_k:=\min\{\lambda_{k-1}^2-\lambda_k^2,\lambda_k^2-\lambda_{k+1}^2\}\ge C(\sigma^2+\sigma\lambda_1)(d_Sd_T/d_G)^{1/2},$$
with the convention that $\lambda_0=\infty$ and $\lambda_{r+1}=0$, then
$$
\max\left\{\widehat{\lambda}_k^2-\lambda_k^2, 1-|\langle\widehat{\bv}_k,\bv_k\rangle|, 1-|\langle\widehat{\bw}_k,\bw_k\rangle|\right\}\le C(\sigma^2+\sigma\lambda_1) g_k^{-1}(d_Sd_T/d_G)^{1/2},
$$
with probability tending to one as $d_G\to\infty$.
\end{thm}

Theorem \ref{th:matricize} indicates that the eigenvalue $\lambda_k$ and its associated eigenvectors $\bv_k$ and $\bw_k$ can be estimated consistently whenever the eigen-gap
$$
g_k\gg  \sigma^2(d_Sd_T/d_G)^{1/2}.
$$
In the context of spatiotemporal expression data, the number of genes $d_G$ is typically much larger than $d_Sd_T$. Therefore, even if the eigen-gap is constant, the spatial and temporal PCA can still be consistently estimated.

\subsubsection{Power iteration}
Although Theorem \ref{th:matricize} suggests that the eigenvalue and eigenvector estimates obtained via our tensor folding scheme is consistent under fairly general conditions, they can actually be further improved. We can indeed use them as the initial value for power iteration or altering least squares to yield estimates that converge to the truth at faster rates. 

Power iteration is perhaps the most commonly used algorithm for tensor decomposation \citep{KoldarBader}. Specifically, let $\bb^{[0]}$ and $\bc^{[0]}$ be initial values for $\bv_k$ and $\bw_k$. Then at the $m$th ($m\ge 1$) iteration, we update $\ba$, $\bb$ and $\bc$ as follows:
%
\begin{enumerate}
\item[$\bullet$] Let $\ba^{[m]}=\ba/\|\ba\|$ where
$$\ba=\bX\times_2\bb^{[m-1]}\times_3\bc^{[m-1]};$$
\item[$\bullet$] Let $\bb^{[m]}=\bb/\|\bb\|$ where
$$\bb=\bX\times_1\ba^{[m]}\times_3\bc^{[m-1]}-\sigma^2\bb^{[m-1]};$$
\item[$\bullet$] Let $\bc^{[m]}=\bc/\|\bc\|$ where
$$\bc=\bX\times_1\ba^{[m]}\times_2\bb^{[m-1]}-\sigma^2\bc^{[m-1]}.$$
\end{enumerate}

The following theorem shows that the algorithm, after a certain number of iterations, yields estimates of the tensor principal components at an optimal convergence rate.

\begin{thm}
\label{th:power}
Let $\bb^{[m]}$ and $\bc^{[m]}$ be the estimates of $\bv_k$ and $\bw_k$ from the $m$th modified power iteration with initial values $\bb^{[0]}=\widehat{\bv}_k$ and $\bc^{[0]}=\widehat{\bw}_k$ obtained by tensor unfolding as described before. Suppose that the conditions of Theorem \ref{th:matricize} hold. Then there exist absolute constants $C_1,C_2>0$ such that if
$$
\lambda_k^2g_k\ge C_1(\sigma^2+\lambda_1\sigma)\lambda_1^2\sqrt{d_Sd_T\over d_G},
$$
then for any
$$
m\ge -C_2\log\left(\lambda_k^{-2}(\sigma^2+\lambda_1\sigma)\sqrt{d_S+d_T\over d_G}\right),
$$
we have
$$
\max\left\{1-|\langle\bb^{[m]},\bv_k\rangle|, 1-|\langle\bc^{[m]},\bw_k\rangle|\right\}=O_p\left(\lambda_k^{-2}(\sigma^2+\lambda_1\sigma)\sqrt{d_S+d_T\over d_G}\right),\quad {\rm as\ }d_G\to\infty.
$$
\end{thm}

Note that we only require that the number of genes $d_G$ diverges in Theorem \ref{th:power}, which is the most relevant setting in spatiotemporal expression data. If the singular values $\lambda_1,\ldots,\lambda_r$ are simple and finite, as typically the case in practice, then Theorem \ref{th:power} indicates that the spatial and temporal PCAs can be estimated  at the rate of convergence $\sqrt{(d_S+d_T)/d_G}$. This is to be compared with the unfolding estimates which converge at the rate of $\sqrt{d_Sd_T/d_G}$.

It is also worth noting, assuming that $\lambda_k$s and $\sigma$ are finite, the rate of convergence given by Theorem \ref{th:power} is optimal in the following sense. Suppose that $\bv_k$ is known in advance, it is not hard to see that $\bX\times_2\bv_k$ is a sufficient statistics for $\bw_k$. Because $\bw_k$ is the usual principal component of $\bX\times_2\bv_k$, following classical theory for principal components \citep[see, e.g.,][]{muirhead2009aspects}, we know that the optimal rate of convergence for estimating $\bw_k$ is of the order $\sqrt{d_T/d_G}$. Similarly, even if $\bw_k$ is known apriori, the optimal rate of convergence for estimating $\bv_k$ would be of the order $\sqrt{d_S/d_G}$. Obviously, not knowing either $\bv_k$ or $\bw_k$ only makes their estimation more difficult. Therefore, the rate of convergence established in Theorem \ref{th:power} is the best attainable.

A key difference between the power iteration described above and the usual ones is that subtract $\sigma^2\bb^{[m-1]}$ and $\sigma^2\bc^{[m-1]}$ when updating $\bb$ and $\bc$ at each iteration. This modification is motivated by a careful examination of the effect of noise $\bE$ on the power iteration. Although not essential for the performance of the final estimate, this adjustment allows for faster convergence of the power iterations. In practice, when $\sigma$ is unknown, one can estimate it by the sample variance of the residual tensor with the initial estimate. A careful inspection of the proof of Theorem \ref{th:power} suggests that the results continue to hold in this case because of the consistency of the initial value.

\section{Numerical Experiments}
\label{sec:sim}

To demonstrate the merits of the tensor PCA method described in the previous section, we conducted several sets of simulations.

\subsection{Estimation accuracy}
We begin with a simple simulation setup designed to investigate the effect of dimensionality and signal strength on the estimation of tensor accuracy. In particular, we simulated data tensor from the following rank one tensor PCA model:
\begin{equation}
\label{eq:rankone}
\bX=\sqrt{d}\lambda \bu\otimes\bv\otimes \bw+\bE.
\end{equation}
To assess the effect of dimensionality, we consider cubic tensors of dimension $\R^{d\times d\times d}$ where $d=25, 50, 100$ or $200$. The principal components $\bv$ and $\bw$, as well as the loadings $\bu$ were uniformly sampled from the unit sphere in $\R^{d}$. We recall that a uniform sample from the unit sphere in $\R^d$ can be obtained by $Z/\|Z\|$ where $Z\sim N(0,I_d)$. The noise tensor $\bE$ is a Gaussian ensemble whose entries are independent standard normal variables.

To assess the effect of signal-to-noise ratio on the quality of our estimates, we set $\lambda=3,4$ or $5$. For each combination of $d$ and $\lambda$, 200 $\bX$s were simulated from model \eqref{eq:rankone}. For each simulated data tensor $\bX$, we computed the estimated principal components both by tensor unfolding (UFD) and by power iteration (PIT) using tensor unfolding for initialization as discussed before. The estimation error was measured by $\max\{1-|\langle\hat{\bv}, \bv\rangle|, 1-|\langle\hat{\bw}, \bw\rangle|\}$. The results, averaged over the 200 runs, are summarized in Figure \ref{fig:sim1}. It is evident from the comparison, power iteration improves the quality of estimates, especially for situations with low signal-to-noise ratio, that is small $\lambda$, or high dimensionality, that is large $d$. These observations are in agreement with the theoretical analysis presented in Theorems \ref{th:matricize} and \ref{th:power}.

\begin{figure}[htbp]
	\centering
	\includegraphics[scale=0.5]{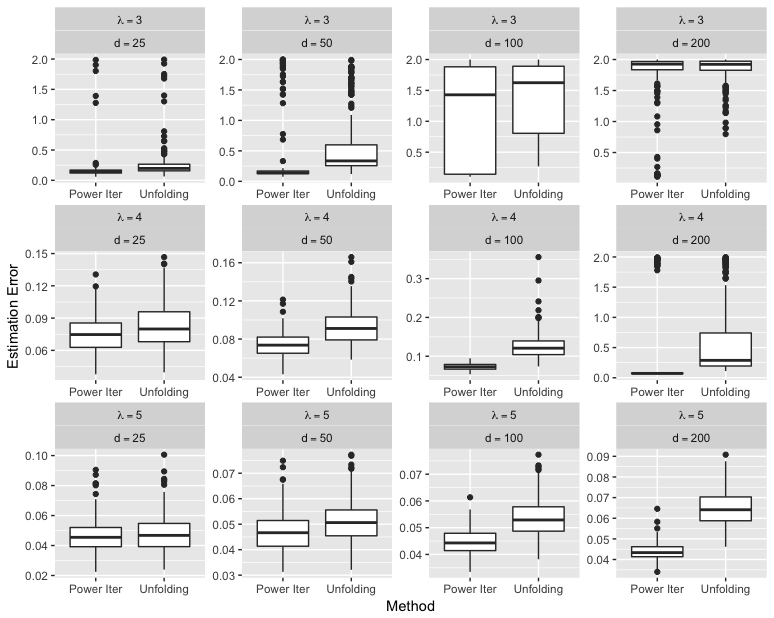}
	\caption{Comparison of estimation error based on tensor unfolding and power iteration with different signal strength ($\lambda$) and dimension ($d$). The boxplots are produced based on 200 simulation runs.}\label{fig:sim1}
\end{figure}

In general, we can see that power iteration can improve the accuracy of estimates based on tensor unfolding. Such an improvement, as suggested by our theoretical development hinges upon the consistency of the unfolding estimates. In the most difficult case when $\lambda=3$ and $d=200$, tensor unfolding fails to provide a consistent estimate of the principal components, and as a result, power iteration also performs poorly. In all other cases, power iteration significantly improves upon the unfolding estimate. The improvement is least significant in the easiest case with $\lambda=5$ and $d=25$ when unfolding estimate already appears to be quite accurate.

To gain further insights into the operating characteristics of the power iteration, we examine how the estimation error changes from iteration to iteration for 50 typical simulation runs with $\lambda=4$ and $d=200$ in Figure \ref{fig:sim1b}. First, it is evident to see the estimation error reduces quickly with the iterations. It is also worth noting that the algorithm converges in only several iterations. This has great practical implication as computation is often a significant issue when dealing with tensor data.

\begin{figure}[htbp]
	\centering
	\includegraphics[scale=0.6]{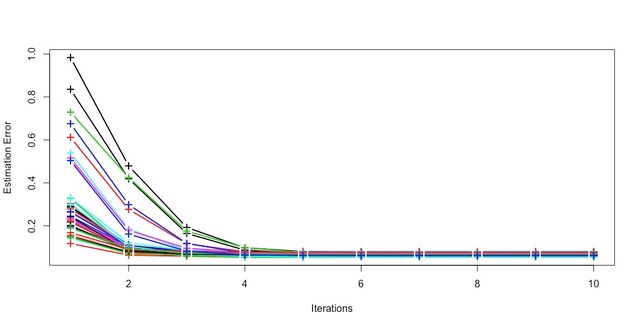}
	\caption{Estimation error as a function of iterations for 50 typical simulated datasets with $\lambda=4$ and $d=200$.}\label{fig:sim1b}
\end{figure}

Our development was motivated by the analysis of spatiotemporal expression data. To better assess the performance of our method in such a context, we now consider a simulation setting designed to mimic it. More specifically, we simulated $2000\times 10\times 13$ data tensors from tensor PCA model of rank four:
%
$$
\bX=\sqrt{d}\lambda\sum_{k=1}^4{5-k\over 4} \cdot\bu_k\otimes\bv_k\otimes \bw_k+\bE,
$$
where we fix $\sigma=1$ and let $\lambda$ vary among $4, 8$ and $16$. The eigenvectors $\bu$, $\bv$ and $\bw$ were uniformly sampled from the Grassmaniann of conformable dimensions. This simulation setting allows us to appreciate the effect of eigengap and eigenvalue, as well as the unequal dimensions on the accuracy of our estimates. We compare the proposed tensor PCA approach with the classical PCA approach for estimating each of the principal component. The results reported in Figure \ref{fig:sim2} again confirms our theoretical findings and suggests the superior performance of the proposed approach over the classical PCA.

\begin{figure}[htbp]
	\centering
	\includegraphics[scale=0.5]{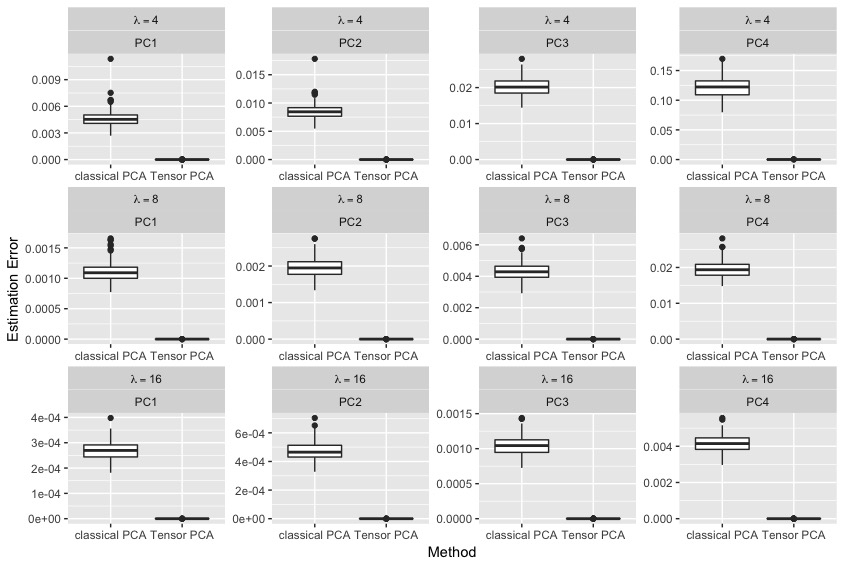}
	\caption{Comparison of estimation error based on classical PCA and the proposed tensor PCA with different signal strength ($\lambda$). The boxplots are produced based on 200 simulation runs.}\label{fig:sim2}
\end{figure}

\subsection{Clustering based on tensor PCA}
Oftentimes in practice, PCA is not the final goal of data analysis. It is commonly used as an initial step to reduce the dimensionality before further analysis. For example, PCA based clustering is often performed when dealing with gene expression data. See, e.g., \cite{yeung01pca}. Similarly, our tensor PCA can serve the same purpose. To investigate the utility of our approach in this capacity, we conducted a set of simulation studies where for each simulated dataset, we first estimated the loadings $\bu_k$s and then applied clustering to the loadings. To fix ideas, we adopted the popular k-means technique for clustering although other alternatives could also be employed. 

Motivated by the dataset from \cite{kang2011spatio} which we shall discuss in further details in the next section, we simulated a data tensor of size $\R^{1087\times 10\times 13}$ from the following model:
\begin{equation}
\label{eq:simmodel}
\bX=\sum_{k=1}^3\lambda_k \bu_k\otimes\bv_k\times\bw_k+\sigma^2\bE.
\end{equation}
where $\lambda_1=337.8$, $\lambda_2=27.1$, $\lambda_3=9.0$, and $\sigma=0.2$. These values, along with the principal components $\bv_k$ and $\bw_k$ are based on estimates when fitting a tensor PCA model to the data from \cite{kang2011spatio}. The clusters, induced by the loadings $\bu_k$, were generated as follows. For a given number $K$ of clusters, we first generated the cluster centroids $C\in\mathbb{R}^{K\times 3}$ from right singular vector matrix of $K$ by $3$ Gaussian random matrix. We then assigned clusters among $1087$ observations and generated the observed tensor with $\sigma=1,5,10,20$, representing different levels of signal-to-noise ratio.

For comparison purposes, we also considered using the classical PCA based approach to reduce the dimensionality. For each method, we took the loadings from the first four directions and then applied k-means to infer the cluster membership. We used adjusted Rand Index as a means of measuring the clustering quality. The results for each method and a variety of combinations of dimension, averaged over 200 runs, are reported in Table \ref{tab:RandIdx}. The results suggest that tensor PCA based clustering is superior to that based on the classical PCA.

\begin{table}[ht]
\label{tab:RandIdx}
\centering
\begin{tabular}  {|c|c|c|}
\hline
noise & classical PCA & tensor PCA \\\hline
20&0.118(0.074)&0.234(0.063)\\\hline
10&0.166(0.073)&0.364(0.072)\\\hline
5&0.242(0.106)&0.659(0.113)\\\hline
1&0.592(0.227)&0.989(0.037)\\\hline
\end{tabular}%
\caption{Clustering performance comparison between the classical  PCA and tensor PCA, in terms of Rand index averaged over 200 simulation runs. Numbers in parentheses are the standard deviations.}
\end{table}

\section{Application to Human Brain Expression Data}
We now turn to the spatiotemporal expression data from \cite{kang2011spatio} that we alluded to earlier.

\label{sec:real}
\subsection{Dataset description and preprocessing}
\cite{kang2011spatio} reported the generation and analysis of exon-level transcriptome and associated genotyping data from multiple brain regions and neocortical areas of developing and adult post-mortem human brains. To characterize the spatiotemporal dynamics of the human brain transcriptome, they created a 15-period system spanning the periods from embryonic development to late adulthood,  as shown in Table \ref{tab:kang}.

\begin{table}
\centering
\begin{tabular}{ccc}
\hline
\hline
\textbf{Period}&\textbf{Description}&\textbf{Age}\\\hline
1&Embryonic&4PCW$\leq$Age$<$8PCW\\
2&Early fetal&8PCW$\leq$Age$<$10PCW\\
3&Early fetal & 10PCW$\leq$Age$<$13PCW\\
4&Early mid-fetal&13PCW$\leq$Age$<$16PCW\\
5&Early mid-fetal&16PCW$\leq$Age$<$19PCW\\
6&Late mid-fetal&19PCW$\leq$Age$<$24PCW\\
7&Late fetal&24PCW$\leq$Age$<$38PCW\\
8&Neonatal and early infancy&0M(birth)$\leq$Age$<$6M\\
9&Late infancy&6M$\leq$Age$<$12M\\
10&Early childhood&1Y$\leq$Age$<$6Y\\
11&Middle and late childhood&6Y$\leq$Age$<$12Y\\
12&Adolescence&12Y$\leq$Age$<$20Y\\
13&Young adulthood&20Y$\leq$Age$<$40Y\\
14&Middle adulthood&40Y$\leq$Age$<$60Y\\
15&Late adulthood&60Y$\leq$Age\\\hline
\end{tabular}
\caption{Periods of human development and adulthood as defined by \cite{kang2011spatio}: M -- postnatal months; PCW -- post-conceptional weeks; Y -- postnatal years.}
\label{tab:kang}
\end{table}

Transient prenatal structures and immature and mature forms were sampled from 16 brain regions, including 11 neocortex (NCX) areas, from multiple specimens per period. In total, the data include 31 males and 26 females with age ranging from 5.7 weeks post-conception to 82 years. Among them, 39 subjects have data from both hemispheres. Except for Periods 1 and 2 as specified in Table \ref{tab:kang}, tissue samples from 16 brain regions were collected, including the cerebellar cortex (CBC), mediodorsal nucleus of the thalamus (MD), striatum (STR), amygdala (AMY), hippocampus (HIP) and 11 areas of the neocortex, including the orbital prefrontal cortex (OFC), dorsolateral prefrontal cortex (DFC), ventrolateral prefrontal cortex (VFC), medial prefrontal cortex (MFC), primary motor cortex (M1C), primary somatosensory cortex (S1C), posterior inferior parietal cortex (IPC), primary auditory cortex (A1C), posterior superior temporal cortex (STC), inferior temporal cortex (ITC) and the primary visual cortex (V1C). Readers are referred to  \cite{kang2011spatio} for more discussion on the sampling location for tissues used in the study.

%

The original dataset contains expression measurement for $17568$ gene, obtained through the Affymetrix GeneChip Human Exon 1.0 ST Array platform. Appropriate normalization and transformation were applied as detailed in \cite{kang2011spatio}. The sample sizes varies among 16 brain regions and 15 time periods. Each of the 57 post-mortem brains was collected at a certain time point of development, expression levels of which were measured across all the regions with several missing values. Periods 1 and 2 were excluded from our analysis because they correspond to embryonic and early fetal development, when most of the 16 brain regions sampled in future periods have not differentiated. Since neocortex regions are quite different from the other 5 regions and we are more interested in neocortex areas, we only included 10 neocortex areas in our analysis, with the exception of V1C because of its location and distinct expression profiles \citep{pletikos2014temporal}.

Following \cite{hawrylycz2015canonical}, we selected genes with reproducible spatial patterns across individuals according to their correlations between samples, leading to a total of $1087$ genes. After taking mean across subjects with same gene, location, and time period we got a data tensor of size $d_G=1087$, $d_S=10$ and $d_T=13$.

\subsection{Analysis based on tensor PCA}
Before applying the tensor PCA, we first centered the gene expression measurements by subtracting the mean expression level for each gene because we are primarily interested in the spatial and temporal dynamics of the expression levels. To remove the mean level, however it is more subtle than the classical PCA, we want to remove both mean spatial effect and mean temporal effect. More specifically, we applied tensor PCA to $\tilde{\bX}\in {\mathbb R}^{d_G\times d_T\times d_S}$ where
$$
\tilde{x}_{gst}=x_{gst}-\bar{x}_{g\cdot t}-\bar{x}_{gs\cdot}+\bar{x}_{g\cdot\cdot}
$$
and $\bX$ is the original data tensor. As in the classical PCA, we can look at the scree plot to examine the contribution of each component in the tensor PCA model. We can see that the contribution from the principal components quickly tapers off. We shall focus on the top three components to fix ideas.

\begin{figure}[htbp]
\centering
\includegraphics[scale=0.7]{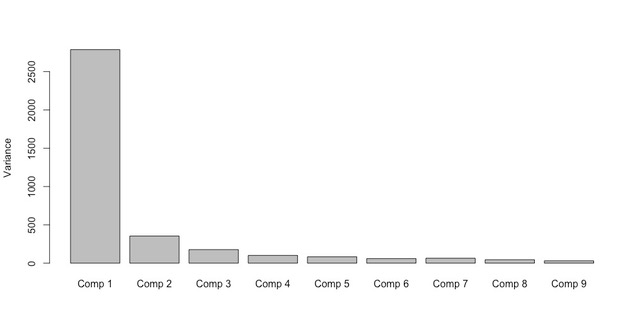}
\caption{Scree plot of the tensor PCA for the dataset from \cite{kang2011spatio}.}
\label{fig:chooseR}
\end{figure}

To gain insights, the top three spatial and temporal principal components are given in Figure \ref{fig:bases3}. And the top three spatial factors are mapped to brain neocortex regions in Figure \ref{fig:regional_3D}, where the color represents value, the darker the higher. It is interesting to note, from the temporal trajectories, that the first two factors show clear signs of prenatal development (until Period 7) while the third factor exhibits increasing influence from young childhood (from Period 11). Factor 1 shows a spatial gradient effect that expression level tapers off from ITC to MFC or the other way. Remarkably, the same effect was reported in \cite{miller2014transcriptional}, which is explained by intrinsic signaling controlled partly by graded expression of transcription factors. Some representative genes such as FGFR3 and CBLN2 were found to preserve in both human and mouse neocortex. Taking temporal effect into consideration, factor 1 indicates that the gradient effect diminishes from early fetal (Period 3) to late fetal (Period 7), and almost vanishes after early infancy. Same effects were observed in \cite{pletikos2014temporal} that areal transcriptional become more synchronized during postnatal development. Factor 2 suggests the importance of prenatal development of M1C and S1C. Both areas are well represented in the second factor while essentially absent from the other factors. This observation based on our analysis seems to agree with recent findings in neuroscience that activation patterns of extremely preterm infants' primary somatosensory cortex area are predictive of future development outcome. See, e.g., \cite{nevalainen14development}. Factor 3 distinguishes middle adulthood (Period 14) and late adulthood (Period 15) with different value in ITC and MFC comparing other 8 regions. This effect was reported in \cite{pletikos2014temporal} that MFC and ITC have much higher number of neocortical interareal differentially expressed (DEX) genes. In term of aging, declining metabolism in MFC correlates with declining cognitive function \citep{pardo2007brain, gutchess2007aging, fjell2009high, donoso2014foundations}, and shrinkage of ITC increases with age \citep{raz2005regional}. When we consider 3 factors together, we can validate the temporal hourglass pattern observed in \cite{pletikos2014temporal} that huge number of DEX genes exist before infancy (Period 8), and areal differences almost vanish from infancy to adulthood (Period 14) and reappear in late adulthood (Period 15). 

To better understand these three factors, we conducted gene set enrichment analysis based on Gene Ontology (http://geneontology.org/) for each factor. We calculated the relative weight of factor $i$ for each gene by $|u_i|/{\sum_{j=1}^3 |u_j|}$, where $u\in \R^3$ is one row of gene factors. For each factor, we chose the top 15\% quantile genes to form the gene sets. The results are presented in Table \ref{tab:enrich_factors}. Factor 1 relates with anatomical structure development, and this result is consistent with its spatial gradient pattern and decrease in magnitude of temporal pattern. Factor 2 has enriched term in sensory organ development, and this agrees with its huge magnitude in S1C. Besides, regulation of anatomical structure morphogenesis term supports the smooth spatial pattern from S1C and M1C to MFC and ITC. Factor 3 is enriched in innervation related with aging \citep{coyle1983alzheimer, lauria1999epidermal}, startle response associated with ITC \citep{sabatinelli2005parallel}, and chemical synaptic transmission related with aging \citep{luebke2004normal}. 

\begin{table}
\centering
\begin{tabular}{c|c|c}
\hline
\hline
Factor&Enriched Term& P-value with Bonferroni Correction\\\hline
\multirow{ 2}{*}{1}&anatomical structure development&4.65E-04\\
&developmental process&2.93E-03\\\hline
\multirow{ 4}{*}{2}&nervous system development&4.20E-04\\
&sensory organ development&1.09E-03\\
&positive regulation of signal transduction&1.36E-02\\
&generation of neurons&1.98E-02\\\hline
\multirow{ 5}{*}{3}&chemical synaptic transmission&3.23E-06\\
&multicellular organismal response to stress&7.32E-04\\
&nucleic acid metabolic process&9.09E-04\\
&ion transmembrane transport&8.02E-04\\
&innervation&1.62E-02\\
&startle response&2.79E-02\\\hline
\end{tabular}
\caption{Gene enrichment analysis results on factors}
\label{tab:enrich_factors}
\end{table}


\begin{figure}[htbp]
\centering
\includegraphics[scale=0.65]{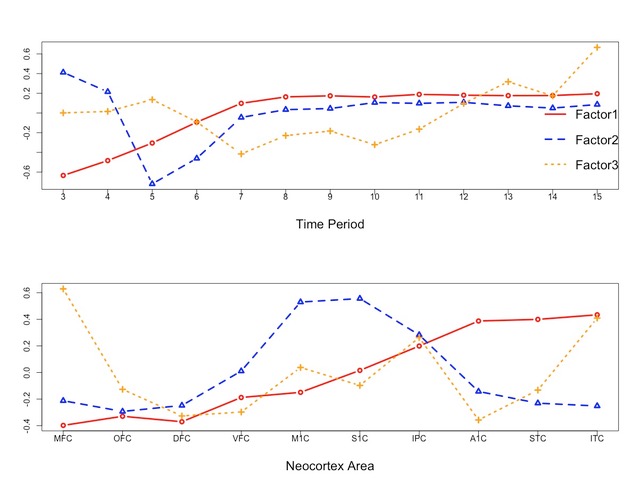}
\caption{Temporal and spatial factors of tensor PCA for the dataset from \cite{kang2011spatio}.}
\label{fig:bases3}
\end{figure}

\begin{figure}[htbp]
\centering
\subfigure{\includegraphics[scale=0.6]{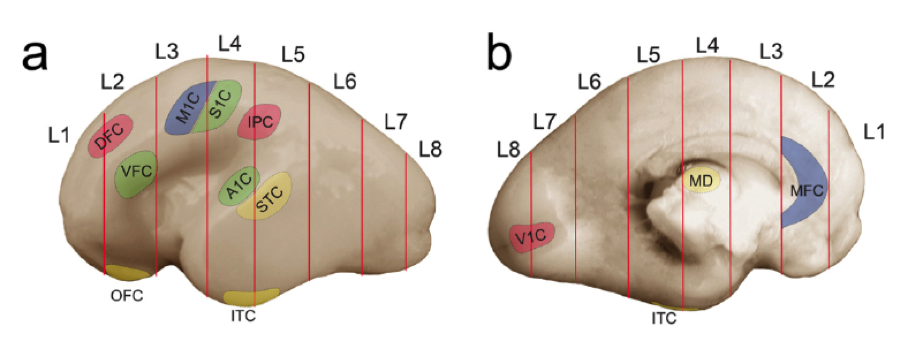}}\\
\subfigure{\includegraphics[scale=0.6]{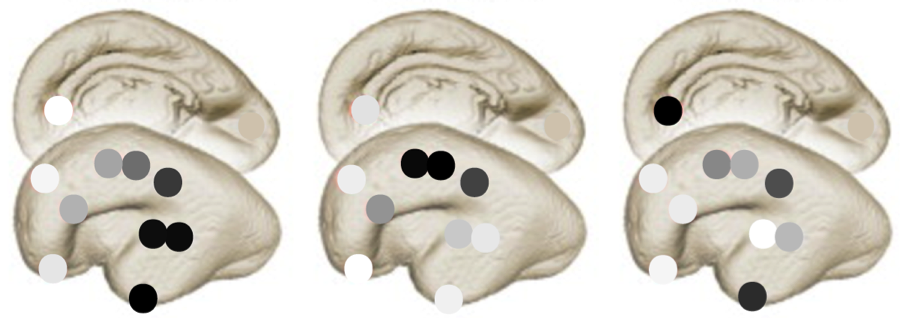}}
\caption{Spatial factors on locations of neocortex from Period 6.}
\label{fig:regional_3D}
\end{figure}


To further examine the meaning of the spatial factors, we use the three spatial factors as the coordinates for each of the 10 locations in a 3D plot as shown in Figure \ref{fig:regionPCA}. Remarkably the spatial patterns of these locations are fairly consistent with the physical locations of these neocortex regions in the brain.

\begin{figure}[htbp]
\centering
\includegraphics[scale=0.4]{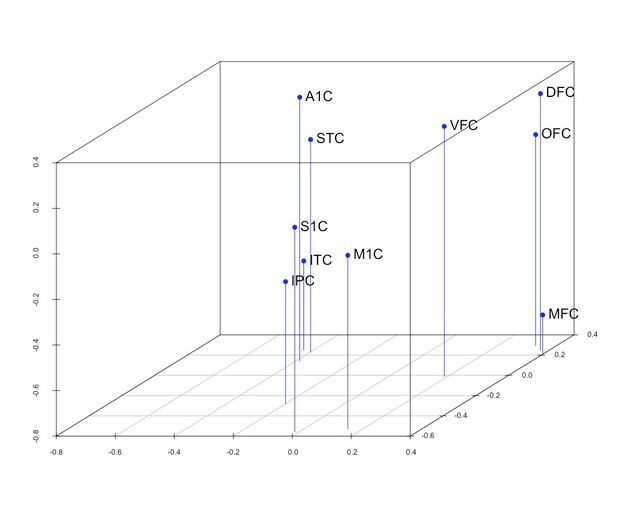}
\caption{Loadings on the top three spatial factors for each of the ten neocortex regions.}
\label{fig:regionPCA}
\end{figure}

Based on three dimensional representation of genes, we identify ten outliers, which are: SLN, GPR64, PROKR2, NEFL, BCL6, GABRQ, DNM1DN3-4, CALB1, PVALB, and VAMP1. GPR64 belongs to G protein-coupled receptors, which underlie the responses to both chemical and mechanical stimuli in olfactory sensory neurons \citep{connelly2015g}. We found that GPR64 achieves peak value in S1C at early mid-fetal (Period 5), which suggests that early mid-fetal may be a critical period for olfactory sensory development. PROKR2 is essential for the regulation of circadian behavior and mice lacking PROKR2 lost precision in timing the onset of nocturnal locomotor activity \citep{prosser2007prokineticin}. It gets peaked in S1C, IPC, A1C, STC, and ITC at Period 5-7, and these areas are associated with receiving and interpreting sensory, auditory processing, and recognizing visual stimuli. These are consistent with PROKR2's functions. Reduced expression of NEFL is observed in anterior cingulate gyrus, motor cortex, and thalamus of autism patients \citep{anitha2012brain}. BCL6 controls neurogenesis \citep{tiberi2012bcl6}. It gets maximum values in S1C and M1C at Period 5, which suggests that the neurogenesis starts earlier in these two regions comparing to others. CALB1 is found to be expressed in certain neuronal subtypes \citep{usoskin2015unbiased}. This may suggest that composition of this type of neuron has a huge spatial and temporal variation. PVALB has been reported to associate with neuropsychiatric disorders including schizophrenia and autism \citep{kaiser2016transgenic}. It has higher expression value in S1C and A1C, which suggests some interneurons expressing PVALB have specific functions related to sensory of S1C and A1C. VAMP1 is the physiologically relevant toxin target in motor neurons \citep{peng2014widespread}, and we indeed observe that it achieves higher value at M1C. 

Finally, we used the factors estimated based on our tensor PCA model as the basis for clustering. In particular, we applied k-means clustering with $k=5$ clusters to the three dimensional factor loadings. The resulting cluster sizes are 156, 167, 332, 280, and 152, respectively. Gene set enrichment analysis based on Gene ontology was performed for each group with the results presented in Table \ref{tab:enrich}. 

\begin{table}
\centering
\begin{tabular}{c|c|c}
\hline
\hline
Cluster&Enriched Term& P-value after Bonferroni Correction\\\hline
\multirow{ 5}{*}{1}&nervous system development&8.58E-11\\
&anatomical structure development&3.43E-09\\
& neurogenesis& 1.63E-05\\
& regulation of developmental process&3.12E-05\\
&cell communication&9.93E-05\\\hline
\multirow{ 5}{*}{2}&chemical synaptic transmission&8.38E-08\\
&inorganic ion transmembrane transport&2.98E-04\\
&nucleic acid metabolic process&6.57E-04\\
&regulation of postsynaptic membrane potential&8.45E-04\\
&multicellular organismal response to stress &1.24E-02\\\hline
\multirow{ 5}{*}{3}&single-organism process&1.81E-10\\
&regulation of localization&9.92E-04\\
&single organism signaling&1.06E-03\\
&response to stimulus&1.59E-03\\
&regulation of multicellular organismal process&6.18E-03\\\hline
\multirow{ 4}{*}{4}&single-organism process &4.13E-06\\
&anatomical structure development&2.72E-04\\
&nervous system development&4.05E-04\\
&signal transduction&4.84E-02\\\hline
\multirow{ 7}{*}{5}&single-organism developmental process &6.18E-05\\
& forebrain development &4.77E-03\\
&chemical synaptic transmission &1.18E-03\\
&neuron projection morphogenesis&9.42E-03\\
&axon development &9.65E-03\\
&regulation of neuron differentiation&3.23E-02\\
&regulation of smooth muscle cell migration &3.88E-02\\\hline
\end{tabular}
\caption{Gene enrichment analysis results}
\label{tab:enrich}
\end{table}

These results show a clear separation among different functional groups. This further indicates that the spatiotemporal pattern of a gene informs its functionality. Moreover, enriched terms such as anatomical structure development, forebrain development are highly associated with the spatial areas of neocortex, which again suggests the the meaningfulness of the tensor principal components.

\section{Conclusions}
\label{sec:dis}
In this paper, we have introduced a generalization of the classical PCA that can be applied to data in the form of tensors. We also proposed efficient algorithms to estimate the principal components using a novel combination of power iteration and tensor unfolding. Both theoretical analysis and numerical experiments point to the efficacy of our method. Although the methodology is generally applicable to other applications, our development was motivated by the analysis of spatiotemporal expression data which in recent years have become a common place in studying brain development among other biological processes. An application of our method to one such example further demonstrates its potential usefulness.

 
\section{Proofs}
\label{sec:proof}
\begin{proof}[Proof of Theorem \ref{th:matricize}]
Write
$$
\bT=\sqrt{d_G}\sum_{k=1}^r \lambda_k \left(\bu_k\otimes\bv_k\otimes \bw_k\right).
$$
Then $\bX=\bT+\bE$. Denote by
$$X_g=(x_{gst})_{1\le s\le d_S, 1\le t\le d_T}.$$
Let $T_g$, $E_g$ be similarly defined. Then
\begin{eqnarray*}
{1\over d_G}\calM(\bX)^\top\calM(\bX)&=&{1\over d_G}\sum_{g=1}^{d_G}\vec(X_g)\otimes \vec(X_g)\\
&=&\calM\left({1\over d_G}\sum_{g=1}^{d_G}X_g\otimes X_g\right)\\
&=&\calM\left({1\over d_G}\sum_{g=1}^{d_G} T_g\otimes T_g +{1\over d_G}\sum_{g=1}^{d_G} E_g\otimes E_g +{1\over d_G}\sum_{g=1}^{d_G} \left(T_g\otimes E_g+E_g\otimes T_g\right)\right).
\end{eqnarray*}
Hereafter, with slight abuse of notation, we use $\calM$ to denote the matricization operator that collapses the first two, and remaining two indices of a fourth order tensor respectively. Observe that
$$T_g=\sqrt{d_G}\sum_{k=1}^r \lambda_k u_{kg} \left(\bv_k\otimes \bw_k\right).$$
Therefore
$$
T_g\otimes T_g = d_G\sum_{k_1,k_2=1}^r \lambda_{k_1}\lambda_{k_2} u_{k_1g}u_{k_2g}\left(\bv_{k_1}\otimes \bw_{k_1}\otimes\bv_{k_2}\otimes \bw_{k_2}\right).
$$
Because of the orthogonality among $\bu_k$s, we get
$$
{1\over d_G}\sum_{g=1}^{d_G} T_g\otimes T_g=\sum_{k=1}^r \lambda_k^2 \left((\bv_{k}\otimes \bw_{k})\otimes(\bv_{k}\otimes \bw_{k})\right).
$$

On the other hand, note that
$$
\calM\left({1\over d_G}\sum_{g=1}^{d_G} E_g\otimes E_g\right)={1\over d_G}\sum_{g=1}^{d_G} \left(\vec(E_g)\otimes \vec(E_g)\right).
$$
In other words, $\calM(d_G^{-1}\sum_{g=1}^{d_G} E_g\otimes E_g)$ is the sample covariance matrix of independent Gaussian vectors
$$
\vec(E_g)\sim N(0, I_{d_S\cdot d_T}), \qquad 1\le g\le d_G.
$$
Therefore, there exists an absolute constant $C_1>0$ such that
$$
\left\|\calM\left({1\over d_G}\sum_{g=1}^{d_G} E_g\otimes E_g\right)-I_{d_S\cdot d_T}\right\|\le C_1\sigma^2\sqrt{d_Sd_T\over d_G}.
$$
with probability tending to one as $d_G\to\infty$. See, e.g., \cite{vershynin12}.

Finally, observe that
$$
\sum_{g=1}^{d_G}T_g\otimes E_g=\sqrt{d_G}\sum_{k=1}^r \lambda_k \left[\bv_k\otimes \bw_k\otimes \left(\sum_{g=1}^{d_G}u_{kg}E_g\right)\right]=:\sqrt{d_G}\sum_{k=1}^r \lambda_k \left(\bv_k\otimes \bw_k\otimes Z_k\right).
$$
By the orthogonality of $\bu_k$s, it is not hard to see that $Z_k$s are independent Gaussian matrices:
$$
\vec(Z_k)\sim N(0,\sigma^2I_{d_S\cdot d_T}),
$$
so that there exists an absolute constant $C_2>0$ such that
$$
\left\|\calM\left({1\over d_G}\sum_{g=1}^{d_G} \left(T_g\otimes E_g+E_g\otimes T_g\right)\right)\right\|\le {2\over d_G}\left\|\calM\left(\sum_{g=1}^{d_G} T_g\otimes E_g\right)\right\|\le C_2\lambda_1\sigma\sqrt{d_Sd_T\over d_G},
$$
with probability tending to one.

To sum up, we get
$$
\left\|{1\over d_G}\calM(\bX)^\top\calM(\bX)-A\right\|\le (C_1\sigma^2+C_2\lambda_1\sigma)\sqrt{d_Sd_T\over d_G}.
$$
where
$$
A=I_{d_S\cdot d_T}+\sum_{k=1}^r \lambda_k^2 \left[\vec\left(\bv_{k}\otimes \bw_{k}\right)\otimes\vec\left(\bv_{k}\otimes \bw_{k}\right)\right].
$$
It is clear that
$$
\left\{(1+\lambda_k^2, \vec(\bv_k\otimes \bw_k)): 1\le k\le r\right\}
$$
are the leading eigenvalue-eigenvector pairs of $A$.

Recall that $(\widehat{\lambda}_k^2, \widehat{\bh}_k)$ is the $k$th eigenvalue-eigenvector pair of $\calM(\bX)^\top\calM(\bX)$. By Lidskii's inequality,
$$
|\widehat{\lambda}_k^2-\lambda_k^2|\le (C_1\sigma^2+C_2\lambda_1\sigma)\sqrt{d_Sd_T\over d_G}.
$$
See, e.g., \cite{Lid1950, katobook}. Then
\begin{eqnarray*}
\|\vec^{-1}(\widehat{\bh}_k)-\bv_k\otimes\bw_k\|^2&\le& \|\vec^{-1}(\widehat{\bh}_k)-\bv_k\otimes\bw_k\|_{\rm F}^2\\
&=&2-2\langle \widehat{\bh}_k,\vec(\bv_k\otimes\bw_k)\rangle\\
&\le&2\left\|\widehat{\bh}_k\otimes \widehat{\bh}_k-\vec(\bv_k\otimes\bw_k)\otimes \vec(\bv_k\otimes\bw_k)\right\| \\
&\le&8(C_1\sigma^2+C_2\lambda_1\sigma)g_k^{-1}\sqrt{d_Sd_T\over d_G},
\end{eqnarray*}
where the last inequality follows from Lemma 1 from \cite{koltchinskii14}. For large enough $C$, we can ensure that
$$
\|\vec^{-1}(\widehat{\bh}_k)-\bv_k\otimes\bw_k\|^2\le 8(C_1\sigma^2+C_2\lambda_1\sigma)g_k^{-1}\sqrt{d_Sd_T\over d_G}\le {1\over 4}.
$$
Recall also that $\widehat{\bv}_k$ and $\widehat{\bw}_k$ be the leading singular vectors of $\vec^{-1}(\widehat{\bh}_k)$. By Wedin's perturbation theorem, we obtain immediately that
$$
\max\left\{1-|\langle\widehat{\bv}_k, \bv_k\rangle|, 1-|\langle\widehat{\bw}_k,\bw_k\rangle|\right\}\le 32(C_1\sigma^2+C_2\lambda_1\sigma)\sigma^2g_k^{-1}\sqrt{d_Sd_T\over d_G}.
$$
See, e.g., \cite{Wedin1972, montanari2014statistical}.
\end{proof}
\vskip 25pt

\begin{proof}[Proof of Theorem \ref{th:power}]
Denote by
$$
\tilde{\bb}= \left({1\over d_G}\sum_{g=1}^{d_G}X_g\otimes X_g\right)\times_2 \bc^{[m-1]}\times_3 \bc^{[m-1]}\times_4 \bb^{[m-1]}-\sigma^2\bb^{[m-1]}.
$$
It is not hard to see that
$$
\bb^{[m]}=\tilde{\bb}/\|\tilde{\bb}\|.
$$
Let $\calM^{-1}$ be the inverse of the matricization operator $\calM$ that unfold a fourth order tensor into matrices, that is, $\calM^{-1}$ reshapes a $(d_Sd_T)\times (d_Sd_T)$ matrix into a fourth order tensor of size $d_S\times d_T\times d_S\times d_T$. Observe that
\begin{eqnarray*}
{1\over d_G}\sum_{g=1}^{d_G}X_g\otimes X_g&=&{1\over d_G}\sum_{g=1}^{d_G} T_g\otimes T_g +{1\over d_G}\sum_{g=1}^{d_G} E_g\otimes E_g +{1\over d_G}\sum_{g=1}^{d_G} \left(T_g\otimes E_g+E_g\otimes T_g\right)\\
&=&\lambda_k^2 \left((\bv_{k}\otimes \bw_{k})\otimes(\bv_{k}\otimes \bw_{k})\right)+\sum_{j\neq k} \lambda_j^2 \left((\bv_{j}\otimes \bw_{j})\otimes(\bv_{j}\otimes \bw_{j})\right)\\
&&+\sigma^2\calM^{-1}(I_{d_S\cdot d_T})+\left({1\over d_G}\sum_{g=1}^{d_G} E_g\otimes E_g -\calM^{-1}(I_{d_S\cdot d_T})\right)\\
&&+{1\over d_G}\sum_{g=1}^{d_G} \left(T_g\otimes E_g+E_g\otimes T_g\right)\\
&=:&\lambda_k^2 \left((\bv_{k}\otimes \bw_{k})\otimes(\bv_{k}\otimes \bw_{k})\right)+\Delta_1+\sigma^2\calM^{-1}(I_{d_S\cdot d_T})+\Delta_2+\Delta_3.
\end{eqnarray*}
We get
$$
\tilde{\bb}=\lambda_k^2 \langle \bb^{[m-1]},\bv_k\rangle\langle \bc^{[m-1]},\bw_k\rangle^2\bv_k+(\Delta_1+\Delta_2+\Delta_3)\times_2 \bc^{[m-1]}\times_3 \bc^{[m-1]}\times_4 \bb^{[m-1]},
$$
where we used the fact that
$$
\calM^{-1}(I_{d_S\cdot d_T})\times_2 \bc^{[m-1]}\times_3 \bc^{[m-1]}\times_4 \bb^{[m-1]}=\bb^{[m-1]}.
$$
Therefore
\begin{eqnarray*}
|\langle \tilde{\bb},\bv_k\rangle|&=&\left|\lambda_k^2 \langle \bb^{[m-1]},\bv_k\rangle\langle \bc^{[m-1]},\bw_k\rangle^2+\langle \Delta_1+\Delta_2+\Delta_3, \bv_k\otimes \bc^{[m-1]}\otimes \bc^{[m-1]}\otimes \bb^{[m-1]}\rangle\right|\\
&=&\lambda_k^2 |\langle \bb^{[m-1]},\bv_k\rangle|\langle \bc^{[m-1]},\bw_k\rangle^2+\left|\langle \Delta_2+\Delta_3, \bv_k\otimes \bc^{[m-1]}\otimes \bc^{[m-1]}\otimes \bb^{[m-1]}\rangle\right|\\
&\ge&\lambda_k^2 |\langle \bb^{[m-1]},\bv_k\rangle|\langle \bc^{[m-1]},\bw_k\rangle^2-\|\Delta_2+\Delta_3\|.
\end{eqnarray*}
Denote by
$$
\tau_m=\min\{|\langle \bb^{[m]},\bv_k\rangle|, |\langle \bc^{[m]},\bw_k\rangle|\}.
$$
Then,
$$
|\langle \tilde{\bb},\bv_k\rangle|\ge \lambda_k^2\tau_{m-1}^3-\|\Delta_2+\Delta_3\|.
$$

On the other hand, note that
\begin{eqnarray*}
\|\tilde{\bb}\|=\langle \tilde{\bb}, \bb^{[m]}\rangle&\le& \lambda_k^2 \langle \bb^{[m-1]},\bv_k\rangle\langle \bc^{[m-1]},\bw_k\rangle^2\langle \bv_k,\bb^{[m]}\rangle\\
&&+\langle \Delta_1+\Delta_2+\Delta_3, \bb^{[m]}\otimes \bc^{[m-1]}\otimes \bc^{[m-1]}\otimes \bb^{[m-1]}\rangle.
\end{eqnarray*}
Write
$$
P_{\bv_k}^{\perp}=I_{d_S}-\bv_k\otimes\bv_k, \qquad {\rm and}\qquad P_{\bw_k}^\perp=(I_{d_T}-\bw_k\otimes\bw_k).
$$
Then
\begin{eqnarray*}
\|\tilde{\bb}\|&=&\lambda_k^2 \langle \bb^{[m-1]},\bv_k\rangle\langle \bc^{[m-1]},\bw_k\rangle^2\langle \bv_k,\bb^{[m]}\rangle\\
&&+\langle \Delta_1, P_{\bv_k}^{\perp}\bb^{[m]}\otimes  P_{\bw_k}^\perp\bc^{[m-1]}\otimes  P_{\bw_k}^\perp\bc^{[m-1]}\otimes P_{\bv_k}^{\perp}\bb^{[m-1]}\rangle\\
&&+\langle \Delta_2+\Delta_3, \bb^{[m]}\otimes \bc^{[m-1]}\otimes \bc^{[m-1]}\otimes \bb^{[m-1]}\rangle\\
&\le&\lambda_k^2 \langle \bb^{[m-1]},\bv_k\rangle\langle \bc^{[m-1]},\bw_k\rangle^2\langle \bv_k,\bb^{[m]}\rangle\\
&&+\lambda_1^2\left(1-\langle \bv_k,\bb^{[m]}\rangle^2\right)^{1/2}\left(1-\langle \bv_k,\bb^{[m-1]}\rangle^2\right)^{1/2}\left(1-\langle \bw_k,\bc^{[m-1]}\rangle^2\right)+\|\Delta_2+\Delta_3\|\\
&\le&\lambda_k^2 |\langle \bb^{[m-1]},\bv_k\rangle|\langle \bc^{[m-1]},\bw_k\rangle^2\\
&&+\lambda_1^2\left(1-\langle \bv_k,\bb^{[m]}\rangle^2\right)^{1/2}\left(1-\langle \bv_k,\bb^{[m-1]}\rangle^2\right)^{1/2}\left(1-\langle \bw_k,\bc^{[m-1]}\rangle^2\right)+\|\Delta_2+\Delta_3\|\\
&\le&\lambda_k^2\tau_{m-1}^3+\lambda_1^2\left(1-\tau_{m-1}^2\right)^{3/2}\left(1-\langle \bv_k,\bb^{[m]}\rangle^2\right)^{1/2}+\|\Delta_2+\Delta_3\|.
\end{eqnarray*}
Therefore,
\begin{eqnarray*}
|\langle\bb^{[m]},\bv_k\rangle|&=&|\langle \tilde{\bb},\bv_k\rangle|/\|\tilde{\bb}\|\\
&\ge& 1-\left(\lambda_k^2\tau_{m-1}^3\right)^{-1}\left[\lambda_1^2\left(1-\tau_{m-1}^2\right)^{3/2}\left(1-\langle \bv_k,\bb^{[m]}\rangle^2\right)^{1/2}\right]\\
&&-\left(\lambda_k^2\tau_{m-1}^3\right)^{-1}\|\Delta_2+\Delta_3\|\\
&\ge&1-4\left(\lambda_k^2\tau_{m-1}^3\right)^{-1}\left[\lambda_1^2\left(1-\tau_{m-1}\right)^{3/2}\left(1-|\langle \bv_k,\bb^{[m]}\rangle|\right)^{1/2}\right]\\
&&-\left(\lambda_k^2\tau_{m-1}^3\right)^{-1}\|\Delta_2+\Delta_3\|\\
&\ge&1-\max\biggl\{8\left(\lambda_k^2\tau_{m-1}^3\right)^{-1}\left[\lambda_1^2\left(1-\tau_{m-1}\right)^{3/2}\left(1-|\langle \bv_k,\bb^{[m]}\rangle|\right)^{1/2}\right], \\
&&2\left(\lambda_k^2\tau_{m-1}^3\right)^{-1}\|\Delta_2+\Delta_3\|\biggr\}\\
&\ge&1- \max\left\{64\left(\lambda_k^2\tau_{m-1}^3\right)^{-2}\lambda_1^4\left(1-\tau_{m-1}\right)^3,2\left(\lambda_k^2\tau_{m-1}^3\right)^{-1}\|\Delta_2+\Delta_3\|\right\}.
\end{eqnarray*}
Assume that
\begin{equation}
\label{eq:taucond}
\tau_{m-1}\ge \max\left\{1-{1\over 64}\left({\lambda_k\over \lambda_1}\right)^2,{1\over 2}\right\},
\end{equation}
which we shall verify later. Then
\begin{equation}
\label{eq:biter}
1-|\langle\bb^{[m]},\bv_k\rangle|\le \max\left\{{1\over 2}\left(1-\tau_{m-1}\right),16\lambda_k^{-2}\|\Delta_2+\Delta_3\|\right\}.
\end{equation}

Similarly, we can show that
$$
1-|\langle\bc^{[m]},\bw_k\rangle|\le\max\left\{{1\over 2}\left(1-\tau_{m-1}\right), 16\lambda_k^{-2}\|\Delta_2+\Delta_3\|\right\}.
$$
Together, they imply that
\begin{equation}
\label{eq:tauiter}
1-\tau_m\le\max\left\{{1\over 2}\left(1-\tau_{m-1}\right), 16\lambda_k^{-2}\|\Delta_2+\Delta_3\|\right\}.
\end{equation}
It is clear from (\ref{eq:tauiter}) that if 
\begin{equation}
\label{eq:finalbd0}
1-\tau_{m-1}\le 16\lambda_k^{-2}\|\Delta_2+\Delta_3\|,
\end{equation}
so is $1-\tau_{m}$. Thus (\ref{eq:finalbd0}) holds for any
$$
m\ge -\log_2\left({16\over 1-\tau_0}\lambda_k^{-2}\|\Delta_2+\Delta_3\|\right).
$$

We now derive bounds for $\|\Delta_2+\Delta_3\|$. By triangular inequality $\|\Delta_2+\Delta_3\|\le \|\Delta_2\|+\|\Delta_3\|$. By Lemma \ref{le:delta2},
$$
\|\Delta_2\|\le 6\sigma^2\sqrt{d_S+d_T\over d_G}.
$$
Next we consider bounding $\|\Delta_3\|$. Recall that
$$
\Delta_3={1\over d_G}\sum_{g=1}^{d_G}T_g\otimes E_g+{1\over d_G}\sum_{g=1}^{d_G}E_g\otimes T_g.
$$
By triangular inequality,
$$
\|\Delta_3\|\le \left\|{1\over d_G}\sum_{g=1}^{d_G}T_g\otimes E_g\right\|+\left\|{1\over d_G}\sum_{g=1}^{d_G}E_g\otimes T_g\right\|={2\over d_G}\left\|\sum_{g=1}^{d_G}T_g\otimes E_g\right\|.
$$
Note that
$$
\sum_{g=1}^{d_G}T_g\otimes E_g=\sqrt{d_G}\sum_{k=1}^r \lambda_k \left[\bv_k\otimes \bw_k\otimes \left(\sum_{g=1}^{d_G}u_{kg}E_g\right)\right]=:\sqrt{d_G}\sum_{k=1}^r \lambda_k \left(\bv_k\otimes \bw_k\otimes Z_k\right),
$$
where $Z_k$s are independent $d_S\times d_T$ Gaussian ensembles. By Lemma \ref{le:delta3}, we get
$$
\left\|\sum_{g=1}^{d_G}T_g\otimes E_g\right\|=O_p\left(\lambda_1\sigma\sqrt{d_G(d_S+d_T)}\right), \quad {\rm as\ }d_G\to\infty,
$$
where we used the fact that $r\le \min\{d_S,d_T\}$. Therefore,
$$
\|\Delta_3\|=O_p\left(\lambda_1\sigma\sqrt{d_S+d_T\over d_G}\right).
$$
Thus, (\ref{eq:finalbd0}) implies that
\begin{equation}
\label{eq:finalbd}
1-\tau_{m}=O_p\left(\lambda_k^{-2}(2\sigma^2+\lambda_1\sigma)\sqrt{d_S+d_T\over d_G}\right),
\end{equation}
for any large enough $m$.

It remains to verify condition (\ref{eq:taucond}), which we shall do by induction. In the light of Theorem \ref{th:matricize} and the assumption on $\lambda_1$ and $\lambda_k$, we know that it is satisfied when $m=0$, as soon as the numerical constant $C>0$ is taken large enough. Now if $\tau_{m-1}$ satisfies (\ref{eq:taucond}), then (\ref{eq:tauiter}) holds. We can then deduct that the lower bound given by (\ref{eq:taucond}) also holds for $\tau_{m}$.
\end{proof}

\bibliographystyle{plainnat}
\bibliography{mybib}

\appendix
\section{Auxiliary Results}
We now derive tail bounds necessary for the proof of Theorem \ref{th:power}.

\begin{lem}
\label{le:delta2}
Let $\bE\in \R^{d_1\times d_2\times d_3}$ ($d_1\ge d_2\ge d_3$) be a third order tensor whose entries $e_{i_1i_2i_3}$ ($1\le i_k\le d_k$) are independently sampled from the standard normal distribution. Write $E_i=(e_{i_1i_2i_3})_{1\le i_2\le d_2, 1\le i_3\le d_3}$ its $i$th $(2,3)$ slice. Then
$$
\left\|{1\over d_1}\sum_{i=1}^{d_1} \left\{E_{i}\otimes E_i-\E\left(E_{i}\otimes E_i\right)\right\}\right\|\le 6\sqrt{d_2+d_3\over d_1}
$$
with probability tending to one as $d_1\to\infty$.
\end{lem}

\begin{proof}[Proof of Lemma \ref{le:delta2}]
For brevity, denote by
$$
\bT_i=E_{i}\otimes E_i-\E\left(E_{i}\otimes E_i\right)
$$
and
$$
\bT={1\over d_1}\sum_{i=1}^{d_1}\bT_i.
$$
Note that $\bT$ is a $d_2\times d_3\times d_3\times d_2$ tensor obeying
$$
T(\omega)=T(\pi_{14}(\omega))=T(\pi_{23}(\omega)), \qquad \forall\omega\in [d_2]\times[d_3]\times[d_3]\times[d_2],
$$
where $\pi_{k_1k_2}$ permutes the $k_1$ and $k_2$ entry of vector. Therefore
$$
\bT=\sup_{\substack{\ba_1,\ba_2\in \R^{d_2}, \bb_1,\bb_2\in \R^{d_3}\\ \|\ba_1\|,\|\ba_2\|,\|\bb_1\|,\|\bb_2\|= 1}}\langle \bT, \ba_1\otimes \bb_1\otimes \bb_2\otimes \ba_2\rangle=\sup_{\substack{\ba\in \R^{d_2}, \bb\in \R^{d_3}\\ \|\ba\|,\|\bb\|= 1}}\langle \bT, \ba\otimes \bb\otimes \bb\otimes \ba\rangle.
$$
Observe that for any $\ba_1,\ba_2\in {\mathbb S}^{d_2-1}$ and $\bb_1,\bb_2\in {\mathbb S}^{d_3-1}$,
\begin{eqnarray*}
&&\left|\langle \bT, \ba_1\otimes \bb_1\otimes \bb_1\otimes \ba_1\rangle-\langle \bT, \ba_2\otimes \bb_2\otimes \bb_2\otimes \ba_2\rangle\right|\\
&\le&\left|\langle \bT, \ba_1\otimes \bb_1\otimes \bb_1\otimes \ba_1\rangle-\langle \bT, \ba_2\otimes \bb_1\otimes \bb_1\otimes \ba_2\rangle\right|\\
&&+\left|\langle \bT, \ba_2\otimes \bb_1\otimes \bb_1\otimes \ba_2\rangle-\langle \bT, \ba_2\otimes \bb_2\otimes \bb_2\otimes \ba_2\rangle\right|\\
&\le&\left|\langle \bT, (\ba_1-\ba_2)\otimes \bb_1\otimes \bb_1\otimes (\ba_1+\ba_2)\rangle\right|\\
&&+\left|\langle \bT, \ba_2\otimes (\bb_1-\bb_2)\otimes (\bb_1+\bb_2)\otimes \ba_2\rangle\right|\\
&\le&2\|\bT\|\left(\|\ba_1-\ba_2\|+\|\bb_1-\bb_2\|\right).
\end{eqnarray*}
In particular, if $\|\ba_1-\ba_2\|,\|\bb_1-\bb_2\|\le 1/8$, then
\begin{equation}
\label{eq:thinning}
\left|\langle \bT, \ba_1\otimes \bb_1\otimes \bb_1\otimes \ba_1\rangle-\langle \bT, \ba_2\otimes \bb_2\otimes \bb_2\otimes \ba_2\rangle\right|\le {1\over 2}\|\bT\|.
\end{equation}
We can find a $1/8$ cover set $\calN_1$ of ${\mathbb S}^{d_2-1}$ such that $|\calN_1|\le 9^{d_2}$. Similarly, let $\calN_2$ be a $1/8$ covering set of ${\mathbb S}^{d_3-1}$ such that $|\calN_2|\le 9^{d_3}$. Then by (\ref{eq:thinning})
$$
\|\bT\|\le \sup_{\ba\in \calN_1,\bb\in \calN_2}\langle \bT, \ba\otimes \bb\otimes \bb\otimes \ba\rangle+{1\over 2}\|\bT\|,
$$
suggesting
$$
\|\bT\|\le 2\sup_{\ba\in \calN_1,\bb\in \calN_2}\langle \bT, \ba\otimes \bb\otimes \bb\otimes \ba\rangle.
$$
Now note that for any $\ba\in \calN_1$ and $\bb\in \calN_2$,
$$
\langle \bT_i, \ba\otimes \bb\otimes \bb\otimes \ba\rangle=\langle E_i,\ba\otimes \bb\rangle^2-\E\langle E_i,\ba\otimes \bb\rangle^2=\langle E_i,\ba\otimes \bb\rangle^2-1\sim \chi^2_1-1.
$$
Therefore
$$
\langle \bT, \ba\otimes \bb\otimes \bb\otimes \ba\rangle\sim {1\over d_1}\chi^2_{d_1}-1.
$$
An application of the $\chi^2$ tail bound from \cite{LauMas98} leads to
$$
\P\left\{\langle \bT, \ba\otimes \bb\otimes \bb\otimes \ba\rangle\ge x\right\}\le \exp(-d_1x^2/4),
$$
for any $x<1$. By union bound,
$$
\P\left\{\sup_{\ba\in \calN_1,\bb\in \calN_2}\langle \bT, \ba\otimes \bb\otimes \bb\otimes \ba\rangle\ge x\right\}\le 9^{d_2+d_3}\exp(-d_1x^2/4),
$$
so that
$$
\|\bT\|\le 6\sqrt{d_2+d_3\over d_1}
$$
with probability tending to one as $d_1\to \infty$.
\end{proof}
\vskip 25pt

\begin{lem}
\label{le:delta3}
Let $\{\bv_1,\ldots,\bv_{d_1}\}$ be an orthonormal basis of $\R^{d_1}$, and $\{\bw_1,\ldots,\bw_{d_2}\}$ an orthonormal basis of $\R^{d_2}$. Let $Z_1,\ldots, Z_r$ be independent $d_3\times d_4$ Gaussian random matrix whose entries are independently drawn from the standard normal distribution. Then for any sequence of nonnegative numbers $\lambda_1,\ldots, \lambda_r\le 1$:
$$
\P\left\{\left\|\sum_{k=1}^r \lambda_k \left(\bv_k\otimes \bw_k\otimes Z_k\right)\right\|\ge \sqrt{d_3}+\sqrt{d_4}+\sqrt{2\log r}+t\right\}\le \exp(-t^2/2).
$$
\end{lem}

\begin{proof}[Proof of Lemma \ref{le:delta3}]
Observe that
\begin{eqnarray*}
\left\|\sum_{k=1}^r \lambda_k \left(\bv_k\otimes \bw_k\otimes Z_k\right)\right\|&=&\sup_{\ba\in {\mathbb S}^{d_1-1},\bb\in {\mathbb S}^{d_2-1}} \left\|\sum_{k=1}^r\lambda_k\langle \ba, \bv_k\rangle\langle \bb, \bw_k\rangle Z_k\right\|\\
&=&\sup_{\ba\in {\mathbb S}^{r-1},\bb\in {\mathbb S}^{r-1}} \left\|\sum_{k=1}^r\lambda_ka_kb_k Z_k\right\|\\
&\le&\sup_{\ba\in {\mathbb S}^{r-1},\bb\in {\mathbb S}^{r-1}} \sum_{k=1}^r\lambda_ka_kb_k \|Z_k\|\\
&\le&\left(\max_{1\le k\le r}\lambda_k\|Z_k\|\right)\left(\sup_{\ba\in {\mathbb S}^{r-1},\bb\in {\mathbb S}^{r-1}}\sum_{k=1}^ra_kb_k\right)\\
&\le&\max_{1\le k\le r}\|Z_k\|.
\end{eqnarray*}
By concentration bounds for Gaussian random matrices,
$$
\P\left\{\|Z_k\|\ge \sqrt{d_3}+\sqrt{d_4}+t\right\}\le \exp(-t^2/2).
$$
See, e.g., \cite{vershynin12}. The desired claim then follows by applying union bound to $\|Z_k\|$s.
\end{proof}

\end{document}